\documentclass[a4paper,preprintnumbers,floatfix,superscriptaddress,aps,10pt,twocolumn,notitlepage,longbibliography,noarxiv]{quantumarticle}
\usepackage[numbers,sort&compress]{natbib}
\usepackage{ltxgrid}
\usepackage{tikz}
\usetikzlibrary{automata, positioning}
\pdfoutput=1
\usepackage[utf8]{inputenc}
\usepackage[english]{babel}
\usepackage[T1]{fontenc}
\usepackage{amsmath}
\usepackage{amssymb}
\usepackage{amsthm}
\usepackage{mathtools}
\usepackage[colorlinks=true,citecolor=blue,urlcolor=magenta,linkcolor=red]{hyperref}
\usepackage{todonotes}
\setuptodonotes{inline}
\usepackage{scalerel}
\usepackage{braket}
\usepackage{subfigure}
\usepackage{accents}
\usepackage{bm,bbm}
\usepackage{tikz}
\usetikzlibrary{automata, positioning, arrows}
\usepackage{lipsum}
\usepackage[capitalize]{cleveref}
\crefname{figure}{Fig.}{Figs.}
\crefname{algorithm}{Protocol}{Protocols}
\usepackage{array} 

\usepackage{listings}
\usepackage{xcolor}

\definecolor{codegreen}{rgb}{0,0.6,0}
\definecolor{codegray}{rgb}{0.5,0.5,0.5}
\definecolor{codepurple}{rgb}{0.58,0,0.82}
\definecolor{backcolour}{rgb}{0.95,0.95,0.92}

\lstdefinestyle{mystyle}{
    backgroundcolor=\color{backcolour},   
    commentstyle=\color{codegreen},
    keywordstyle=\color{magenta},
    numberstyle=\tiny\color{codegray},
    stringstyle=\color{codepurple},
    basicstyle=\ttfamily\footnotesize,
    breakatwhitespace=false,         
    breaklines=true,                 
    captionpos=b,                    
    keepspaces=true,                 
    numbers=left,                    
    numbersep=5pt,                  
    showspaces=false,                
    showstringspaces=false,
    showtabs=false,                  
    tabsize=2
}

\DeclarePairedDelimiterX\ketbra[2]{\vert}{\vert}%
  {#1\kern0.15ex\delimsize\rangle\delimsize\langle\kern0.15ex\mathopen{}#2}
\newcommand{\kb}[2]{\ketbra{#1}{#2}}

\newcommand{\inFid}{\mathrm{inFid}}
\newcommand{\pfail}{p_\mathrm{fail}}
\newcommand{\one}{\mathbbm{1}}

\newcommand{\iS}{\mathrm{s}}

\newcommand{\iH}{\mathrm{h}}

\newcommand{\iT}{\mathrm{t}}

\newcommand{\CC}{\mathbb{C}}
\newcommand{\RR}{\mathbb{R}}
\newcommand{\Znn}{\mathbb{Z}_{\geq 0}} 
\newcommand{\Zp}{\mathbb{Z}_{>0}} 

\renewcommand{\H}{\mathcal{H}}

\newcommand{\ii}{\ensuremath\mathrm{i}}
\newcommand{\e}{\ensuremath\mathrm{e}} 

\newcommand{\veps}{\varepsilon}

\renewcommand{\Pr}{\mathbb{P}} 
\newcommand{\T}{\intercal} 

\newcommand{\vbt}[1]{{\ttfamily #1}} 

\newcommand{\argdot}{\,\cdot\,}

\newcommand{\y}{\mathrm{y}} 
\newcommand{\n}{\mathrm{n}} 
\renewcommand{\a}{\mathrm{a}} 
\renewcommand{\r}{\mathrm{r}} 
\newcommand{\A}{\mathcal{A}} 
\newcommand{\eps}{\epsilon} 
\newcommand{\M}{\mathcal{M}} 
\newcommand{\QM}{\mathcal{Q}} 
\newcommand{\LL}{\mathrm{L}} 

\newcommand{\ceil}[1]{\lceil #1 \rceil}
\newcommand{\imax}{{i_{\mathrm{max}}}}

\let\Set\undefined
\providecommand\given{}
\newcommand\SetSymbol[1][]{%
  \nonscript\:#1\vert
  \allowbreak
  \nonscript\:
  \mathopen{}}
\DeclarePairedDelimiterX\Set[1]\{\}{%
  \renewcommand\given{\SetSymbol[\delimsize]}
  #1
}
\DeclarePairedDelimiterX{\abs}[1]{\lvert}{\rvert}{%
  \ifblank{#1}{\,\cdot\,}{#1}
}

\newtheorem{corollary}{Corollary}
\newtheorem{definition}{Definition}
\newtheorem{observation}{Observation}

\newtheorem{result}{Result}
\newtheorem{example}{Example}
\newtheorem{proposition}{Proposition}

\crefname{definition}{Def.}{Defs.}
\crefname{section}{Sec.}{Secs.}
\crefname{appendix}{Appendix}{Appendix} 

\definecolor{anna}{rgb}{0.0, 0.6, 0.0}

\definecolor{jan}{rgb}{0.6, 0, 1}

\definecolor{mari}{rgb}{1.0, 0.0, 0.22}

\definecolor{lucas}{rgb}{0.0, 0.4, 0.6}

\usepackage{acronym}
\usepackage{algorithm}
\usepackage[noend]{algorithmic}
\floatname{algorithm}{Protocol}

\newacro{DFA}{deterministic finite automaton}
\newacro{pvDFA}{promise-version deterministic finite automaton}
\newacro{pvPFA}{promise-version probabilistic finite automaton}
\newacro{PFA}{probabilistic finite automaton}
\newacro{QFA}{quantum finite automaton}
\newacro{RB}{randomized benchmarking}
\newacro{GST}{gate set tomography}
\newacro{POVM}{positive operator-valued measure}
\newacro{PVM}{projection-valued measure}
\newacro{CP}{completely positive}
\newacro{CPTP}{completely positive trace preserving}
\newacro{PSD}{positive semidefinite}
\newacro{NISQ}{noisy intermediate-scale quantum}
\newacro{SPAM}{state preparation and measurement} 
\newacro{ONB}{orthonormal basis}
\newacro{SDI}{semi-device independent}
\newacro{QSQ}{quantum system quizzing}
\newacro{MUBs}{mutually unbiased bases}

\DeclareMathOperator{\Tr}{Tr} 

\begin{document}

\title{Certifying Quantum Gates via Automata Advantage}
\author{Anna Schroeder}
\affiliation{Merck KGaA, Darmstadt, Germany}
\affiliation{Department of Computer Science, Technical University of Darmstadt, Darmstadt, Germany}
\email{anna.schroeder@merckgroup.com}

\author{Lucas B. Vieira}
\affiliation{Department of Computer Science, Technical University of Darmstadt, Darmstadt, Germany}

\author{Jan N\"oller}
\affiliation{Department of Computer Science, Technical University of Darmstadt, Darmstadt, Germany}

\author{Nikolai Miklin}
\affiliation{Institute for Quantum-Inspired and Quantum Optimization, Hamburg University of Technology, Germany}
\affiliation{Institute for Applied Physics, Technical University of Darmstadt, Darmstadt, Germany}

\author{Mariami Gachechiladze}
\affiliation{Department of Computer Science, Technical University of Darmstadt, Darmstadt, Germany}
\maketitle

\begin{abstract}
There is growing interest in developing rigorous tests of quantumness that are feasible even before practical quantum advantages become a reality.  Such tests not only aim to certify the quantum nature of a system but also serve as benchmarks for precise quantum control.
In this work, we argue that promise problems, studied in the theory of finite automata, provide a natural framework for designing sound tests of quantum gate quality. 
Soundness, the property that only implementations of sufficiently high quality can pass the test, is a central requirement for meaningful certification. We study several promise problems relevant to quantum gate testing and establish separations between the memory resources required by quantum and classical finite automata to solve them. These separations form the theoretical basis for using promise problems as tests of quantumness.
Finally, we show how results from automata theory, in particular the minimality of automata, can be used to derive soundness guarantees.
\end{abstract}

\section{Introduction}
Quantum computing offers significant advantages in the resources required to perform computational tasks, with time being the most commonly considered resource~\cite{mermin2007quantum}. 
However, this is not the only form of savings that quantum mechanics can offer. 
In particular, encoding information in the state of a quantum system can reduce the amount of working memory needed during computation~\cite{rosset2018resource, giarmatzi2021witnessing, heinosaari2024simple, vieira2024entanglement}.

In the current era of \ac{NISQ} devices, where harnessing the full potential of quantum computing to solve practical problems remains out of reach, there is an ongoing race to demonstrate quantum computational advantage~\cite{arute2019quantum,madsen2022quantum,morvan2024phase,kretschmer2025demonstrating}. 
The central theoretical objective here is to design \emph{tests of quantumness} that can, under a set of easily acceptable assumptions, verify that a given computation was indeed performed on a quantum device.
Another objective of such tests is to certify that the implemented quantum operations, such as logical quantum gates, are of sufficiently high quality. 
In the latter case, it is sufficient to perform smaller-scale computations that demonstrate quantum advantage, e.g., in terms of memory savings, without necessarily surpassing the capabilities of today's classical computers.

Recently, some of us proposed a method for testing the implementation of quantum gates by executing a small number of quantum circuits whose outputs are expected to be deterministic~\cite{noller2025classical,noller2025sound}.
In contrast to other approaches in widespread use~\cite{emerson2005scalable,levi2007efficient, emerson2007symmetrized,merkel2013self,blume2013robust,nielsen2021gate}, the proposed method can be rigorously proven to provide a \emph{sound} certification of quantum gates, meaning that implementations of insufficient quality are guaranteed to be detected.
The key ingredient for achieving soundness is a restriction on the amount of memory accessible to the computer.
This assumption is crucial, as a classical computer with sufficiently large memory can reproduce the results of any quantum computation, given that there are no further restrictions such as time.

In this work, we revisit the type of computational tasks on which the sound certification of Refs.~\cite{noller2025classical,noller2025sound} is based and show that these are instances of the so-called \emph{promise problems} solved with finite-state automata.
Promise problems, which generalize the concept of decision problems in computer science~\cite{even1984complexity,goldreich2006promise,zheng2017promise}, are known to represent a class of problems where quantum mechanics provides an advantage in terms of the memory required for computation~\cite{ambainis2012superiority,ambainis2021automata, yakaryilmaz2009languages,bhatia2019quantum,gruska1999quantum, bianchi2014complexity, qiu2016quantum}. 
Interestingly, many celebrated quantum algorithms, such as Deutsch-Jozsa~\cite{deutsch1992rapid} and Simon's~\cite {simon1997power}, whose complexity is provably smaller than any classical analogs, also solve instances of promise problems.

This paper is organized as follows.
In Section~\ref{sec:preliminaries}, we introduce the necessary concepts from automata theory and argue that promise problems provide a natural framework for designing sound tests for quantum logical operations.
We then formalize, as a promise problem, the task used in Ref.~\cite{noller2025classical} to test a single-qubit phase gate, and discuss several generalizations of it in Section~\ref{sec:separation}.
For each case, we prove a separation in the amount of memory required for classical and quantum automata to solve the problem.
As further examples of promise problems, we consider those that naturally arise in the context of operations on quantum states involving a finite group of gates, such as the single-qubit Clifford group.

When analyzing the minimal amount of memory required by a classical automaton to solve a given promise problem, in some cases, one is also able to argue about its uniqueness.
This, in turn, has implications for the certification of logical operations in a quantum automaton solving the same problem, which we discuss in Section~\ref{sec:uniqueness}.
An important generalization of promise problems involves their restricted versions, where the length of the input is finite, a necessary condition for making the corresponding quantum gate tests applicable in practice, discussed in Section~\ref{sec:QSQ}.
Finally, we examine how the choice of this input length influences the performance of the certification test.
Section~\ref{sec:summary} presents the summary and outlook, outlining the future potential of promise problem tests in quantum information processing tasks.
Some technical details are left for the Appendix.

\section{Preliminaries}
\label{sec:preliminaries}
We begin by introducing the notation used throughout this paper.
Let $\Znn$ and $\Zp$ denote the sets of non-negative and positive integers, respectively.
For any $n \in \Zp$, we use the shorthand notation $[n] = \Set{1, 2, \dots, n}$.
By $\Sigma$ we denote a finite set of symbols, referred to as the \emph{alphabet}.
We use $\Sigma^\ast$ to denote the set of all finite-length strings, called \emph{words}, formed from symbols in $\Sigma$.
The empty word is denoted by $\eps$.
For a symbol $\sigma \in \Sigma$ and $n\in \Znn$, we write $\sigma^n$ for the word consisting of $n$ repetitions of $\sigma$, with the convention that $\sigma^0 = \eps$.
We primarily follow Ref.~\cite{zheng2017promise} in introducing automata theory concepts needed for this paper.

\begin{definition}\label{def:promise_problem}
    A promise problem over $\Sigma$ is a pair $(A_\y,A_\n)$ of disjoint subsets of $\Sigma^\ast$.
\end{definition}
We refer to $A_\y$ and $A_\n$ as the ``yes'' and ``no'' languages, respectively. 
In the promise problem formulation of language recognition, the objective is to decide correctly on inputs from $A_\y \cup A_\n$, in contrast to standard decision problems, where $A_\n=\Sigma^\ast\setminus A_\y$.

In this work, we consider finite automata as the classical model of computation for solving promise problems.
\begin{definition}\label{def:pvDFA}
    A \ac{pvDFA} is specified by a $6$-tuple 
    \begin{equation}
        \A = (S, \Sigma, \delta, s_0, S_\a, S_\r),
    \end{equation}
    where $S$ is a finite set of states, $\Sigma$ is the alphabet, $\delta: S\times \Sigma \to S$ is a transition function, $s_0 \in S$ is the initial state, and $S_\a$ and $S_\r$ are disjoint subsets of $S$, called accept and reject states, respectively.
\end{definition}
\Cref{def:pvDFA} generalizes the concept of \ac{DFA} commonly studied in the computer science literature~\cite{hopcroft2001introduction}.
For a standard \ac{DFA}, we have $S_\r=S\setminus S_\a$.
The transition function of a \ac{DFA} is commonly specified by a state diagram, which is a directed graph with vertices representing states and arrows representing the transitions.
In the case of \ac{pvDFA}, apart from the accept states, depicted by double circles, we also need to specify the reject states, which we draw as filled circles (see \cref{fig:sautomaton}, (Right)). Empty vertices correspond to unlabeled states.

It is common to extend the definition of the transition function to $\delta: S\times \Sigma^\ast\to S$ in a recursive manner $\delta(s,w\sigma) = \delta(\delta(s,w),\sigma)$ for $s\in S$ and $w\in\Sigma^\ast$.

\begin{definition}\label{def:DFA_solve}
    We say that a promise problem $A = (A_\y, A_\n)$ is solved by a \ac{pvDFA} $\A$ if for every $w \in A_\y \cup A_\n $ we have,
    \begin{equation}\begin{split}\label{eq:def_DFA_solve}
        w \in A_\y &\Rightarrow \delta(s_0, w) \in S_\a,\\
        w \in A_\n  &\Rightarrow \delta(s_0, w) \in S_\r.
    \end{split}\end{equation}
\end{definition}
Next, we define the quantum counterpart to \ac{pvDFA}.
\begin{definition}\label{def:QFA}
Let $\H$ be a finite-dimensional Hilbert space.
A \ac{QFA} $\M$ is specified by a $6$-tuple
\begin{equation}\label{eq:defQFA}
    \mathcal{M} = (Q, \Sigma, \{U_\sigma\}_{\sigma \in \Sigma}, \ket{\psi_0}, Q_\a,Q_\r) ,
\end{equation}
where $Q$ is an orthonormal basis of $\H$, called the set of states, $\Sigma$ is the alphabet, $U_\sigma$ is a unitary operator on $\H$ for each $\sigma\in \Sigma$, $\ket{\psi_0}\in \H$ is the initial state and, $Q_\a$ and $Q_\r$ are disjoint subsets of $Q$, called accept and reject states, respectively.
\end{definition}
Various types of \acp{QFA} have been proposed in the past~\cite{kondacs1997power,moore2000quantum,brodsky1999characterizations, ambainis2006algebraic, li2009characterizations, zheng2017promise}.
The definition given above corresponds to the automaton commonly known as the promise version measure-once quantum finite automaton.
A small distinction, however, is that we do not consider additional symbols that are added at the beginning and the end of an input word, which allows a \ac{QFA} to perform additional unitary transformations. 
Instead, we do not require that $\ket{\psi_0}\in Q$, which makes our definition equivalent to the one in Ref.~\cite{zheng2017promise}.
In \ac{QFA} literature, it is common to consider only pure states and unitary transformations, which will also be sufficient for most of the discussions in this work. 
To avoid ambiguity with the global phase, by a quantum state we will mean an equivalence class $\alpha\ket{\psi}$ for $\alpha\in \CC$ with $\abs{\alpha}=1$, with $\ket{\psi}$ being its representative.

The evolution of a \ac{QFA} state is discrete, as in the case of a \ac{DFA}.
When the automaton reads a symbol $\sigma \in \Sigma$ while in state $\ket{\psi}$, it updates its state by applying the corresponding unitary transformation $U_\sigma$, resulting in the new state $U_\sigma \ket{\psi}$.
However, an important difference with \acp{DFA} is that it might be that $U_\sigma\ket{\psi}\notin Q$. 
For this reason, when we talk about a \ac{QFA} accepting or rejecting a word $w$, we need to talk about the probabilities of these events. 
These probabilities are given by the Born rule,
\begin{equation}\label{eq:def_pr_accept}
    \Pr[\M \text{ accepts } w ] =
    \sum_{\ket{\psi}\in Q_\a}\abs{\bra{\psi}U_w\ket{\psi_0}}^2,
\end{equation}
and similarly for $\Pr[\M \text{ rejects } w ]$ when the summation is taken over $Q_\r$.
In \cref{eq:def_pr_accept} we introduced a notation $U_w$, which we also define recursively $U_{w\sigma} = U_\sigma U_w$ for $w\in \Sigma^\ast$ and $\sigma\in \Sigma$.
Finally, we can define what it means for a \ac{QFA} to solve a promise problem.
\begin{definition}\label{def:QFA_solve}
    A promise problem $A = (A_\y,A_\n)$ is solved by a \ac{QFA} $\M$ with an error probability $\veps$ if for every $w\in A_\y\cup A_\n$ we have,
    \begin{equation}\begin{split}\label{eq:def_QFA_solve}
        w \in A_\y &\Rightarrow \Pr[\M \text{ accepts } w ]\geq 1-\veps,\\
        w \in A_\n  &\Rightarrow \Pr[\M \text{ rejects } w ]\geq 1-\veps.
    \end{split}\end{equation}
    When the error $\veps$ is not explicitly mentioned, it is assumed to be $0$.
\end{definition}
The central notion of the following section is the \emph{minimal automaton} for a given promise problem, defined as a \ac{pvDFA} (or \ac{QFA}) with the smallest number of states that solves the problem.

\begin{figure}[t!]
    \centering
    \includegraphics[width=0.8\linewidth]{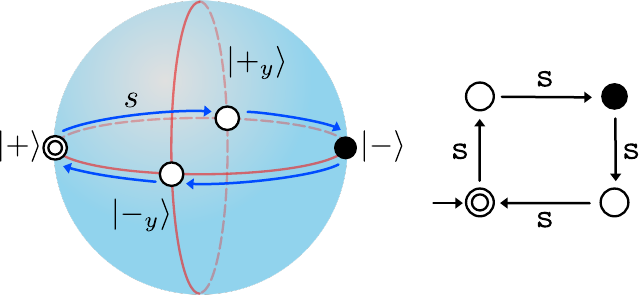}
    \caption{(Left) A qubit \ac{QFA} solving $\mathtt{EO}^1$. (Right) The state diagram of a \ac{pvDFA} solving the same promise problem, which requires $4$ states. The initial state is marked by an incoming arrow, the accept state by a double circle, and the reject state by a filled circle.}
    \label{fig:sautomaton}
\end{figure}

\section{Separation between classical and quantum automata on promise problems}~\label{sec:separation}
It is known that the minimal number of states needed to solve a promise problem can differ for \acp{QFA} and \acp{DFA}~\cite{ambainis2012superiority,ambainis2021automata, yakaryilmaz2009languages,bhatia2019quantum,gruska1999quantum, bianchi2014complexity, qiu2016quantum}.
Here, we focus on promise problems relevant for the certification of quantum gates.
In Ref.~\cite{noller2025sound}, some of us have shown that an implementation of a single-qubit phase gate $\sqrt{Z} = \kb{0}{0}+\ii\kb{1}{1}$ by a quantum computer can be certified if it outputs \vbt{accept} for words $\{\eps,\iS^4\}$ and \vbt{reject} for $\{\iS^2\}$, where $\Sigma=\{\iS\}$, and we assume its memory to be a qubit. 
The reader can recognize that the above sequences are instances of a promise problem $(A_\y,A_\n)$, with 
\begin{equation}\begin{split}\label{eq:EO_s}
    A_\y& =\Set{\iS^{2i}\given i\in \Znn, i\equiv 0\bmod 2},\\
    A_\n& =\Set{\iS^{2i}\given i\in \Znn, i\equiv 1\bmod 2},
\end{split}\end{equation}
and a \ac{QFA} $(\Set{\ket{+},\ket{-}},\Set{\iS},\Set{\sqrt{Z}},\ket{+},\Set{\ket{+}},\Set{\ket{-}})$ can solve it, where $\ket{\pm} = (\ket{0}\pm\ket{1})/\sqrt{2}$.

We begin by reviewing a more general form of this promise problem, which has been independently studied in the literature on finite automata.
\begin{example}\label{ex:EO}
    Let $k\in \Zp$ and $\Sigma=\Set{\sigma}$.
    Consider a promise problem $\mathtt{EO}^k = (\mathtt{EO}_\y^k, \mathtt{EO}_\n^k)$, defined by
    \begin{equation}\begin{split}\label{eq:EO_def}
    \mathtt{EO}_\y^k &= \Set{\sigma^{i2^k} \given i \in\Znn,  i\equiv 0\bmod 2},\\
    \mathtt{EO}_\n^k &= \Set{\sigma^{i2^k} \given i \in\Znn,  i\equiv 1\bmod 2}.
\end{split}\end{equation}
\end{example}
The above promise problem was first introduced in Ref.~\cite{ambainis2012superiority}, and as one can see, the promise problem in \cref{eq:EO_s} corresponds to the special case $k=1$.
Surprisingly, a qubit is still enough to solve this more general problem, while the number of states of a \ac{pvDFA} must grow \emph{exponentially} in $k$.
\begin{result}[\cite{ambainis2012superiority}]\label{res:EO}
There is a \ac{QFA} with $2$ states that solves $\mathtt{EO}^k$.
A minimal \ac{pvDFA} that solves $\mathtt{EO}^k$ has $2^{k+1}$ states.
\end{result}
We only need one modification (except for taking $\Sigma=\Set{\sigma})$ to the \ac{QFA} solving the promise problem in \cref{eq:EO_s}, namely, we need to take $U_\sigma = \sqrt[2^k]{Z} = \kb{0}{0}+\e^{\frac{\pi\ii}{2^k}}\kb{1}{1}$.
A proof of a minimal \ac{pvDFA} can be found in Ref.~\cite{ambainis2012superiority} or~\cite{bianchi2014complexity}, and we also give it in \cref{app:EO_proof} for completeness.
The \acp{DFA} state diagrams together with the Bloch-sphere representations of the quantum state evolutions of the \acp{QFA} for $k=1$ and $k=2$ are given in \cref{fig:sautomaton} and \cref{fig:tautomaton}, respectively.
There, we use symbols $\iS$ and $\iT$, as these cases correspond to $\sqrt{Z}$ and $\sqrt[4]{Z}$ gates (often referred to as $S$- and $T$-gates), which are relevant in quantum computing applications.
\begin{figure}
    \centering
    \includegraphics[width=0.9\linewidth]{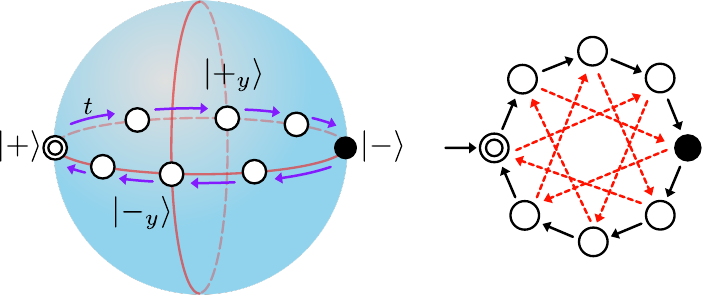}
    \caption{(Left) A qubit \ac{QFA} solving $\mathtt{EO}^2$. (Right) The state diagram of a \ac{pvDFA} solving the same promise problem, which requires $8$ states. Red dashed lines indicate an alternative transition function for the same set of states.}
    \label{fig:tautomaton}
\end{figure}

Next, we discuss an extension of $\mathtt{EO}^1$ to non-unary alphabets.
\begin{example}\label{ex:DIOF}
    Let $k\in \Zp$ and $\abs{\Sigma}=k$. Consider a promise problem $\mathtt{DIOF}^k=(\mathtt{DIOF}^k_\y,\mathtt{DIOF}^k_\n)$, defined by
    \begin{equation}
    \begin{split}
       \mathtt{DIOF}^k_\y &= \Set*{ w \in \Sigma^\ast \given \sum_{\sigma\in\Sigma}a_\sigma\abs{w}_{\sigma} \equiv 0 \bmod{2^{k+1}}}, \\
        \mathtt{DIOF}^k_\n &= \Set*{ w \in \Sigma^\ast \given \sum_{\sigma\in\Sigma}a_\sigma\abs{w}_{\sigma} \equiv 2^k \bmod{2^{k+1}}},
        \end{split}
    \end{equation}
    where to each symbol we associate a weight $a_\sigma$, with all weights forming a tuple $(1,2^1,\dots, 2^{k-1})$, and $\abs{w}_\sigma = \abs{\Set{w_i=\sigma \given i\in [\ell]}}$ denotes the number of symbols $\sigma$ in a word $w=w_1w_2\dots w_\ell$.
\end{example}
The case $k=1$ corresponds to $\mathtt{EO}^1$. 
In Ref.~\cite{noller2025classical}, some of us have shown that $\sqrt{Z}$ and $\sqrt[4]{Z}$ gates can be certified (as part of a single-qubit universal gate set) from the ability of a quantum computer to correctly output \vbt{accept} and \vbt{reject} to a certain set of input words.
The framework of promise problems allows us to identify that the words that need to be tested are instances of $\mathtt{DIOF}^2$.
Indeed, take $\Sigma=\Set{\iS,\iT}$, $a_\iT=1$ and $a_\iS=2$, then a qubit \ac{QFA} with the initial and the accept state $\ket{+}$, the reject state $\ket{-}$, and $U_\iS=\sqrt{Z}$, $U_\iT=\sqrt[4]{Z}$ can solve it.
For a \ac{pvDFA} to solve it, it needs $8$ states, with the state diagram given in~\cref{fig:diofst}.
\begin{figure}[t!]
    \centering
    \includegraphics[width=0.4\linewidth]{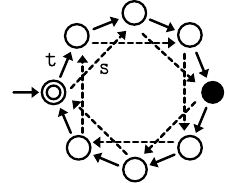}
    \caption{The state diagram of a \ac{pvDFA} solving $\mathtt{DIOF}^{2}$ with $\Sigma=\{\iS,\iT\}$. The solid arrows indicate transitions for $\iT$, while the dashed arrows represent transitions for $\iS$.}
    \label{fig:diofst}
\end{figure}

The separation between \acp{QFA} and \acp{DFA} in the number of states required to solve $\mathtt{DIOF}^k$ was studied in Ref.~\cite{bianchi2014complexity}.
We restate this result below using our notation and definitions.
\begin{result}[\cite{bianchi2014complexity}]
    There is a \ac{QFA} with $2$ states that solves $\mathtt{DIOF}^k$. A minimal \ac{pvDFA} that solves $\mathtt{DIOF}^k$ has $2^{k+1}$ states.
\end{result}
It is easiest to describe a proof if we take $\Sigma=[k]$.
A \ac{QFA} with $Q_\a=\Set{\ket{+}}$, $Q_\r=\Set{\ket{-}}$, the initial state $\ket{+}$ and the unitaries $U_j = \kb{0}{0}+\e^{\frac{\pi \ii}{2^{k-j+1}}}\kb{1}{1}$ for $j\in [k]$ can solve $\mathtt{DIOF}^k$.
Note that in Ref.~\cite{bianchi2014complexity}, the authors choose to give an example of a \ac{QFA} with $\abs{Q}=3$, because they consider a more general variant of the promise problem $\mathtt{DIOF}^k$.
A proof of a minimal \ac{pvDFA} can be found in Ref.~\cite{bianchi2014complexity}.
We also provide a short proof below.
\begin{proof}
    If we consider $\mathtt{EO}^k$ defined over the alphabet $\Sigma=\Set{1}$, then $\mathtt{EO}^k_\y\subset\mathtt{DIOF}^k_\y$ and $\mathtt{EO}^k_\n\subset\mathtt{DIOF}^k_\n$, which means that a \ac{pvDFA} which solves $\mathtt{DIOF}^k$ can also solve $\mathtt{EO}^k$, and thus due to \cref{res:EO}, we must have $\abs{S}\geq 2^{k+1}$.
    Now, we need to show that this number of states suffices.
    Take a \ac{pvDFA} with $S=\Set{s_i}_{i=1}^{2^{k+1}}$ and the transition function $\delta(s_i,1)=s_{i+1}$ for $i\in [2^{k+1}-1]$ and $\delta(s_{2^{k+1}},1)=s_1$ that solves $\mathtt{EO}^k$.
    Now extend the transition function to other symbols $j\in [k]$ according to their weight $a_j=2^{j-1}$: $\delta(s_i,j)=s_{i+2^{j-1}}$, if $i+2^{j-1}\leq 2^{k+1}$, and $\delta(s_i,j)=s_{i+2^{j-1}-2^{k+1}}$, otherwise.
\end{proof}

The next example of a promise problem that we consider arises in the certification of a single-qubit Clifford gate set.
In particular, consider a set of gates $\Set{\sqrt{Z},H}$, where $H = \kb{+}{0}+\kb{-}{1}$ is the Hadamard gate.
In Refs.~\cite{noller2025classical,noller2025sound}, some of us identified a finite set of words sufficient to certify these gates when acting on the initial state $\ket{+}$. 
To the best of our knowledge, no promise problem associated with this case has been defined in the literature.
One possible reason is that this problem does not appear to admit a concise formulation of its accept and reject languages, as was possible for $\mathtt{EO}^{k}$ and $\mathtt{DIOF}^{k}$.

\begin{example}\label{ex:Cl}
    Let $\Sigma=\Set{\iS,\iH}$ and $U_\iS=\sqrt{Z}$, $U_\iH=H$. Consider a promise problem $\mathtt{Cl}=(\mathtt{Cl}_\y, \mathtt{Cl}_\n)$, where
    \begin{equation}\begin{split}\label{eq:def_Cl}
        \mathtt{Cl}_\y &= \Set*{w\in \Sigma^\ast \given \abs{\bra{+}U_w\ket{+}}=1},\\
        \mathtt{Cl}_\n &= \Set*{w\in \Sigma^\ast \given \abs{\bra{-}U_w\ket{+}}=1}.
    \end{split}\end{equation}
\end{example}

We refer to such promise problems, defined on the group structure generated by quantum gates acting on a specific initial state, as quantum-inspired promise problems.

\begin{result}\label{def:sh}
    There is a \ac{QFA} with $2$ states that solves $\mathtt{Cl}$. A minimal \ac{pvDFA} that solves $\mathtt{Cl}$ has $6$ states.
\end{result}
By construction of \cref{ex:Cl}, a qubit \ac{QFA} $(\Set{\ket{+},\ket{-}},\Set{\iS,\iH},\Set{\sqrt{Z},H},\ket{+},\Set{\ket{+}},\Set{\ket{-}})$ solves $\mathtt{Cl}$.
To construct a minimal \ac{pvDFA}, we rely on the following general observation.
\begin{observation}\label{obs:map}
    Consider a promise problem $A$ defined over $\Sigma$ and solved by a \ac{QFA} with the initial state $\ket{\psi_0}$ and transformations $\Set{U_\sigma}_{\sigma\in\Sigma}$. If the set of quantum states $\Set{U_w\ket{\psi_0}\given w\in \Sigma^\ast}$ is finite, then there exists a \ac{pvDFA} solving $A$ with the set of states isomorphic to that set.
\end{observation}
\begin{figure}[t]
\centering
\includegraphics[width=0.5\textwidth]{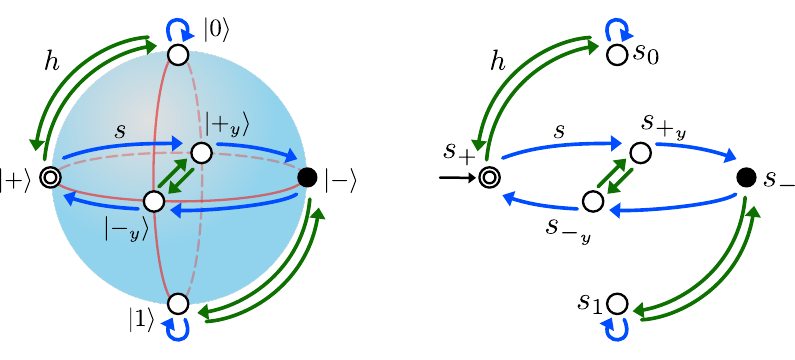}
   \caption{
(Left) A qubit \ac{QFA} solving $\mathtt{Cl}$. (Right) The state diagram of a \ac{pvDFA} solving $\mathtt{Cl}$, which needs $6$ states.}
   \label{fig:shautomaton}
\end{figure}

In \cref{fig:shautomaton}, we give the state diagram of a \ac{pvDFA} solving $\mathtt{Cl}$ with the set of states $S = \Set{s_0,s_1,s_+,s_-,s_{+_y},s_{-_y}}$.
For a proof that this \ac{pvDFA} is minimal, see \cref{app:DFA_sh}. 

\subsection{Generalized promise problems}
In the first part of this section, we observed that promise problems offer a natural framework for testing sets of quantum gates. However, in quantum computing experiments, we often have access to more than just two possible outcomes of a computation. Therefore, it is meaningful to move beyond the standard definitions presented in the Preliminaries section, where promise problems were limited to ``yes''/``no'' languages.
We refer to an ordered set that generalizes $(\y,\n)$ for promise problems and $(\a,\r)$ for automata as a set of \emph{labels}, using the same set and term in both cases.
Since extending \cref{def:promise_problem,def:pvDFA,def:DFA_solve,def:QFA,def:QFA_solve} to label sets containing more than two elements is straightforward, we omit these definitions here.

The first example that we consider generalizes $\mathtt{EO}^k$ to the multi-copy case.
\begin{example}\label{ex:NEO}
    Let $\Sigma = \Set{\sigma_1,\sigma_2, \dots, \sigma_N}$ and the set of labels $\LL=(b)_{b\in\{0,1\}^N}$. Consider a promise problem $\mathtt{N\text{-}EO}^k = (\mathtt{N\text{-}EO}^k_b)_{b\in\LL}$, defined by
    \begin{equation}\label{eq:def_NEO}
        \mathtt{N\text{-}EO}^k_b = \Set{w\in \Sigma^\ast \given \abs{w}_{\sigma_j} \equiv 2^kb_j \bmod 2^{k+1}, j\in [N]},
    \end{equation}
    where $\abs{w}_{\sigma_j} = \abs{\Set{w_i=\sigma_j \given i\in [\ell]}}$ denotes the number of symbols $\sigma_j$ in a word $w=w_1w_2\dots w_\ell$.
\end{example}
Clearly, for $N = 1$, we recover the promise problem $\mathtt{EO}^k$ from \cref{ex:EO}. We show that the exponentially growing separation between the numbers of states in \acp{pvDFA} and the \acp{QFA} solving this promise problem remains to hold, but also the base of the exponent grows exponentially with $N$.

\begin{result}\label{res:NEO}
    There is a \ac{QFA} with $2^N$ states that solves $\mathtt{N\text{-}EO}^k$. A minimal \ac{pvDFA} that solves $\mathtt{N\text{-}EO}^k$ has $2^{N(k+1)}$ states.
\end{result}
Clearly, a multi-qubit generalization of the \ac{QFA} from \cref{ex:EO} with $Q_b= \{\bigotimes_{i=1}^N(\ket{0}+(-1)^{b_i}\ket{1})/\sqrt{2}\}$ for $b\in \LL$, $Q = \bigcup_{b\in \LL} Q_b$, and $U_{\sigma_i} = \sqrt[2^k]{Z}^{(i)}$, which is the unitary $\sqrt[2^k]{Z}$ acting on $i$-th qubit, solves $\mathtt{N\text{-}EO}^k$.
For a proof of a minimal \ac{pvDFA} see \cref{app:multiqubit}.

Our next example is motivated by the dynamics of a single quantum system with $\dim(\H)=q$, where $q\geq 2$. Below, we use the common notation $\omega = \e^{\frac{2\pi\ii}{q}}$ for the $q$-th root of unity, and $\{\ket{k}\}_{k=0}^{q-1}$ for the computational basis.
\begin{example}\label{ex:GEO}
    Let $q$ be a prime number, $r\in \Zp$, $\Sigma=\Set{\sigma}$, and take the set of labels to be $\LL = (0,1,\dots,q-1)$. Consider a promise problem $\mathtt{GEO}^r_q = (\mathtt{GEO}^r_{q,j})_{j\in\LL}$, defined by
    \begin{equation}\label{eq:def_GEO}
       \mathtt{GEO}^r_{q,j} = \Set{\sigma^{i\cdot  r} \given i\in\Znn, i \equiv j \bmod q}, 
    \end{equation}
    for $j \in \LL$.
\end{example}
For $q=r=2$, this problem reduces to $\mathtt{EO}^1$.

\begin{result}\label{res:GEO}
    There is a \ac{QFA} with $q$ states that solves $\mathtt{GEO}^r_q$. A minimal \ac{pvDFA} that solves $\mathtt{GEO}^r_q$ has $q^{m_q + 1}$ states,  where $m_q$ is the highest power of $q$, such that $q^{m_q}$ divides $r$.
\end{result}
An example of a \ac{QFA} that solves $\mathtt{GEO}^r_q$ has $Q = \Set{\ket{x_j}}_{j=0}^{q-1}$, with 
$\ket{x_j}=\frac{1}{\sqrt{q}}\sum_{k=0}^{q-1}\omega^{j k}\ket{k},$
for $j\in \{0,1\dots, q-1\}$. The starting state of the \ac{QFA} is $\ket{x_0}$, and the transitions are given by $U_\sigma=\sqrt[r]{Z_q}=\sum_{k=0}^{q-1} \omega^{\frac{k}{r}}\kb{k}{k}$.
The sets of labeled states are $Q_j=\{\ket{x_j}\}$ for $j\in\LL$. 
For a proof of a minimal \ac{pvDFA} see Appendix~\ref{app:GEORQ}.

\section{Uniqueness of automata solving promise problems}~\label{sec:uniqueness}
In the previous section, we observed that promise problems emerge naturally in the context of testing sets of quantum gates from deterministic outcomes, following the framework of Refs.~\cite{noller2025classical,noller2025sound}.
Here, we deepen this connection by examining the uniqueness of a \ac{pvDFA} and a \ac{QFA} that solve a given promise problem.
Such uniqueness, modulo certain allowable degrees of freedom, plays a central role in proving the soundness of the certification test~\cite{noller2025classical,noller2025sound}.

We look again at the promise problem $\mathtt{EO}^k$. As shown in \cref{res:EO}, any minimal \ac{pvDFA} solving this problem requires $2^{k+1}$ states. However, there is more one can say about the structure of \acp{pvDFA} with $2^{k+1}$ states that solve $\mathtt{EO}^k$.
In the proof of \cref{res:EO}, we establish that the state diagram of every minimal \ac{pvDFA} must form a loop of size $2^{k+1}$, with the initial state being accept and the $2^k$-th state being reject.
Consequently, given any \ac{pvDFA} whose state diagram is such a loop, the state diagram of any other minimal \ac{pvDFA} solving $\mathtt{EO}^k$ (with one accept and one reject state) can be obtained from it by permutation of the unlabeled states, i.e., states which are not the accept or reject one (see, e.g.,~\cref{fig:tautomaton}).
We can use this fact to establish the uniqueness of a \ac{QFA} solving $\mathtt{EO}^k$.


\begin{result}\label{res:QFA_equiv} 
Any \ac{QFA} with $2$ states that solves the promise problem $\mathtt{EO}^k$ is equivalent, up to a choice of \ac{ONB} in $\CC^2$, to the one where the initial state is $\ket{+}$, the sets of accept and reject states are $\Set{\ket{+}}$ and $\Set{\ket{-}}$, respectively, and the transition is realized by a unitary $\kb{0}{0}+\e^{\frac{\pi\ii j}{2^{k}}}\kb{1}{1}$ with $j$ odd integer. 
\end{result}
\begin{proof}
Let $\Sigma=\Set{\sigma}$, $\ket{\psi_0}$ be the initial state of a \ac{QFA} that solves $\mathtt{EO}^k$, and let us denote $\ket{\psi_k}=U_\sigma^k\ket{\psi_0}$ for $k\in \Zp$.
By \cref{def:QFA}, we must have at least one accept and one reject state, which are orthonormal, and for $\abs{Q}=2$, it implies that there are no more accept or reject states. 
Due to \cref{obs:map}, we can always map a \ac{QFA} to a \ac{DFA} with $\abs{\Set{U_w\ket{\psi_0}\given w\in \Sigma^\ast}}$ number of states as long as this number is finite. 
This implies that the set $\{\ket{\psi_k}\}_{k\in \Znn}$ cannot be smaller than $2^{k+1}$, otherwise \cref{res:EO} would not hold.
At the same time, since $\{\sigma^0,\sigma^{2^{k+1}}\}\subset \mathtt{EO}^k_\y$, we must have that $Q_\a = \{\ket{\psi_0}\}$, and $\ket{\psi_{2^{k+1}}}$ is the same state as $\ket{\psi_0}$, which essentially means that $\abs{\Set{U_w\ket{\psi_0}\given w\in \Sigma^\ast}}=2^{k+1}$. 
From the uniqueness of \ac{pvDFA} with $2^{k+1}$ states solving $\mathtt{EO}^k$ we can now infer that $Q_\r = \Set{\ket{\psi_{2^k}}}$.
Since the accept and reject states are orthonormal, we can infer that $U^{2^k}_\sigma = \kb{\psi_{2^k}}{\psi_0}+\e^{\ii\alpha}\kb{\psi_{0}}{\psi_{2^k}}$, for some $\alpha\in \RR$, where we chose the global phase to be $0$. 
With an appropriate change of \ac{ONB} given by a unitary $V=\kb{+}{\psi_0}+\e^{\ii\frac{\alpha}{2}}\kb{-}{\psi_{2^k}}$, we can obtain $V\ket{\psi_0}=\ket{+}$, $V\ket{\psi_{2^k}} = \ket{-}$, and $VU^{2^k}_\sigma V^\dagger = Z$, which implies that $VU_\sigma V^\dagger =\kb{0}{0}+\e^{\frac{\pi\ii j}{2^{k}}}\kb{1}{1}$ for $j$ odd integer.
\end{proof}

For $k = 1$, \cref{res:QFA_equiv} implies that there are two possible \acp{QFA} solving $\mathtt{EO}^1$, up to a choice of basis, with the unitary $U_\sigma$ being either $\sqrt{Z}$ or $\sqrt{Z}^\dagger$. These two implementations are equivalent under complex conjugation, which is another natural degree of freedom in quantum mechanics~\cite{wigner1931gruppentheorie}. Thus, it is possible to certify the implementation of the gate $\sqrt{Z}$ on the initial state $\ket{+}$ in the sense of Ref.~\cite{noller2025classical,noller2025sound}.
However, already for $k = 2$, \cref{res:QFA_equiv} admits four possible transition unitaries $U_\sigma \in \{\sqrt[4]{Z}, \sqrt[4]{Z}^\dagger, Z\sqrt[4]{Z}, \sqrt[4]{Z}^\dagger Z\}$. The complex conjugation degree of freedom reduces this number to two, and no certification in the strict sense is possible without going beyond the framework of promise problems. In general, the freedom of choosing the transition unitary in \cref{res:QFA_equiv} (choosing $j$) corresponds to the freedom of choosing the $2^k$-th root of the $Z$ operator, which we fixed at the beginning of this paper. This is also precisely the degree of freedom that occurs in the uniqueness of the minimal \ac{pvDFA} that solves $\mathtt{EO}^k$.

Arguably, a shorter proof of \cref{res:QFA_equiv} can be constructed by analyzing certain words in $\mathtt{EO}^k$, without relying on the uniqueness of $2^{k+1}$-state \ac{pvDFA}. However, for a more general promise problem with a non-unary alphabet, designing a proof of the uniqueness of \ac{QFA} is a tedious task, and it comprises a large part of the proof of soundness in testing sets of quantum gates~\cite{noller2025classical,noller2025sound}. The connection to \acp{pvDFA} solving the same promise problem provides, in that regard, a more systematic approach.

\section{Applications to certification of quantum gates}\label{sec:QSQ}
To apply the promise-problem framework to testing quantum gates in practice, one must analyze the impact of testing only words of finite length.
For a given promise problem $A = (A_\y,A_\n)$, we call its \emph{restriction} a promise problem $A_\ell = (A_{\y,\ell},A_{\n,\ell})$, where $A_{\y,\ell} = \Set{w\in A_\y\given \abs{w}\leq \ell}$, and $A_{\n,\ell} = \Set{w\in A_\n\given \abs{w}\leq \ell}$, where $\ell$ is the maximal length of the tested  

In Refs.~\cite{noller2025classical,noller2025sound}, some of us have shown that, in certain cases, a finite promise problem can be used to establish the soundness of gate certification under the assumption of known Hilbert space dimension.
This section explores more systematically the choices of the restrictions of promise problems and their effects on gate certification.  

First, we show that the separation between \acp{pvDFA} and \acp{QFA} persists for the restricted promise problem. In this case, however, the required number of states of the minimal solutions depends on the maximal length of the input words. In the following, we discuss the restrictions of $\mathtt{EO}^k$.

\begin{example}\label{ex:EOimax}
    A restricted $\mathtt{EO}^{k}$ is given as $\mathtt{EO}^{k}_\imax = (\mathtt{EO}^{k}_{\imax,\y}, \mathtt{EO}^{k}_{\imax,\n})$, where 
    \begin{equation}\begin{split}
        \mathtt{EO}^{k}_{\imax,\y} &= \{ \sigma^{i2^k} \;|\; i \in \mathbb{Z}_{\geq 0,\,\mathrm{even}} \text{ and } i \leq \imax\},
        \\
       \mathtt{EO}^{k}_{\imax,\n} &= \{ \sigma^{i2^k} \;|\; i \in \mathbb{Z}_{\geq 0,\,\mathrm{odd} } \text{ and } i \leq \imax\}.
    \end{split}\end{equation}
\end{example}

\begin{result}\label{theorem:finitepp}
    A minimal \ac{pvDFA} solving $\mathtt{EO}^{k}_{\imax}$ has $\min\{(\imax+1),2^{k+1}\}$ states.
\end{result}
See \cref{app:KLEO} for a proof. Clearly, the qubit \ac{QFA} that solves $\mathtt{EO}^{k}$ can also solve the finite version of it. 

Now we analyze different choices of restrictions for the quantum gate certification task, for which we will follow the framework of Ref.~\cite{noller2025classical,noller2025sound} for sound gate certification, referred to as \ac{QSQ}.
Unlike the \ac{QFA} formalism, the \ac{QSQ} protocol does not assume perfect (noiseless) state preparation, transformations, or measurements defined on a Hilbert space of a fixed dimension. Instead, QSQ aims to \emph{certify the closeness} of an implemented (noisy) quantum model
\begin{equation}\label{eq:qmodel}
    \QM = (\rho, \{\Lambda_\sigma\}_{\sigma\in\Sigma}, \{M_\a\}_{a \in \LL})
\end{equation}
to a target noiseless quantum model in a chosen distance measure, taking into account the (anti-)unitary degree of freedom inherent to quantum theory~\cite{wigner1931gruppentheorie}.
Here, $\rho$ denotes an initial quantum state, $\{\Lambda_\sigma\}_{\sigma\in\Sigma}$ is a set of quantum channels, and $\{M_\a\}_{a \in \LL}$ are measurement operators given as a \ac{POVM}, with $\LL$ the set of labels in the promise problem.
For models with unitary channels, we use the corresponding unitary operators in the definition of a quantum model in \cref{eq:qmodel}.

\begin{definition}~\label{def:modelequivalence}
Two fixed-dimensional quantum models $\QM = (\rho,\{\Lambda_\sigma\}_{\sigma\in \Sigma},\{M_\a\}_{a\in \LL})$ and $\tilde{\QM} = (\tilde\rho,\{\tilde\Lambda_\sigma\}_{\sigma\in \Sigma},\{\tilde M_\a\}_{a\in \LL})$ are equivalent if, up to a complex conjugation applied to all components, there exists a unitary $U$ such that $U\tilde\rho U^\dagger = \rho$, $U\tilde M_\a U^\dagger = M_\a$ for all $a \in \LL$ and $ U\tilde\Lambda_\sigma(U^\dagger (\argdot) U)U^\dagger=\Lambda_\sigma(\argdot)$ for all $\sigma\in \Sigma$. 
\end{definition}

In the following, we focus our analysis of promise problems in the \ac{QSQ} protocol to instances of $\mathtt{EO}^1_\imax$ for different lengths $\imax$. For each instance, we conduct an in-depth noise robustness analysis where we connect the infidelity between an implemented and the target quantum models with their expected probability of failing the promise problem.

\subsection{Robustness analysis for QSQ}

We define infidelity between two $d$-dimensional quantum models $\QM$ and  $\tilde{\QM}$ as 
\begin{equation}\label{eq:inFid}
	\inFid = \min_{U} \max \{\inFid^U_\rho,\, \inFid^U_\Lambda,\, \mathrm{dist}^U_M\}, 
\end{equation}
where the state infidelity is given by  $\inFid^U_\rho= 1 - \Tr[U \rho U^\dagger \tilde{\rho}]$, the measurement by the operator norm, $\mathrm{dist}^U_M = \max_a \| U M_a U^\dagger - \tilde{M}_a \|_\infty$, and the channel by the average gate infidelity~\cite{kliesch2021theory},
\begin{equation}\label{eq:inFid_gate}
	\inFid^U_\text{gate} = \tfrac{d}{d+1}\left( 1 - \tfrac{1}{d^2} \Tr[\mathcal{C}_{\Lambda_U} {\mathcal{C}_{\tilde{\Lambda}}}^\dagger] \right),
\end{equation}
where $\mathcal{C}_{\Lambda}$ is an unnormalized Choi state of the channel $\Lambda$, and $\Lambda_U$ is a channel $\Lambda$, and $U$ is an (anti-)unitary transformation.

In order to ensure the advantage of a quantum model over a classical one, we compare their probabilities of success under a fixed dimension. 
As the classical counterpart to \ac{QFA} we consider a probabilistic version of \ac{pvDFA}, called 
\ac{pvPFA}.
A \ac{pvPFA} is defined analogously to a \ac{pvDFA}, with the difference that transitions between its states are represented by $T = \{T_\sigma \}_{\sigma \in \Sigma}$, a set of $|S| \times |S|$ column-stochastic matrices, i.e., such that $[T_\sigma]_{ij} \ge 0$ and $\sum_{i} [T_\sigma]_{ij} = 1$ for all $j \in S$, and the starting state is chosen according to a probability distribution $\pi$. The probability of the word $w = \sigma_1 \dots \sigma_\ell$, sampled from a promise problem, producing the answer $a \in \LL$ can then be written as
\begin{equation}\label{eq:pfa_accprob}
\Pr[a|w] = m_a T_{\sigma_\ell} \cdots T_{\sigma_1} \pi
\end{equation}
where $m_a$ is a row vector indicating the states which are labeled $a$, with $[m_a]_i = 1$ if $i \in S_a$, and zero otherwise. In quantum models, this probability is calculated by Born's rule,
\begin{equation}\label{eq:qfa_accprob}
\Pr[a|w] =\Tr[ M_a \Lambda_{\sigma_\ell} \circ\dots \circ \Lambda_{\sigma_1}(\rho) ].
\end{equation}

For a restricted promise problem $A = (A_a)_{a\in \LL}$, we define the probability of failure
\begin{equation}\label{eq:pfail}
    \pfail = 1 - \sum_{a\in\LL}\sum_{w \in A_a} \mu(w) \Pr[a|w],
\end{equation}
where $\mu$ is the sampling distribution over all the words in  the \ac{QSQ} protocol.
In what follows, we will denote by $p_C=\min_{\text{pvPFA}} \pfail$, where the minimization is taken over all possible \acp{pvPFA} with a fixed number of states.

\begin{figure*}
\centering
\begin{minipage}[b]{0.49\textwidth}
\includegraphics[width=\textwidth]{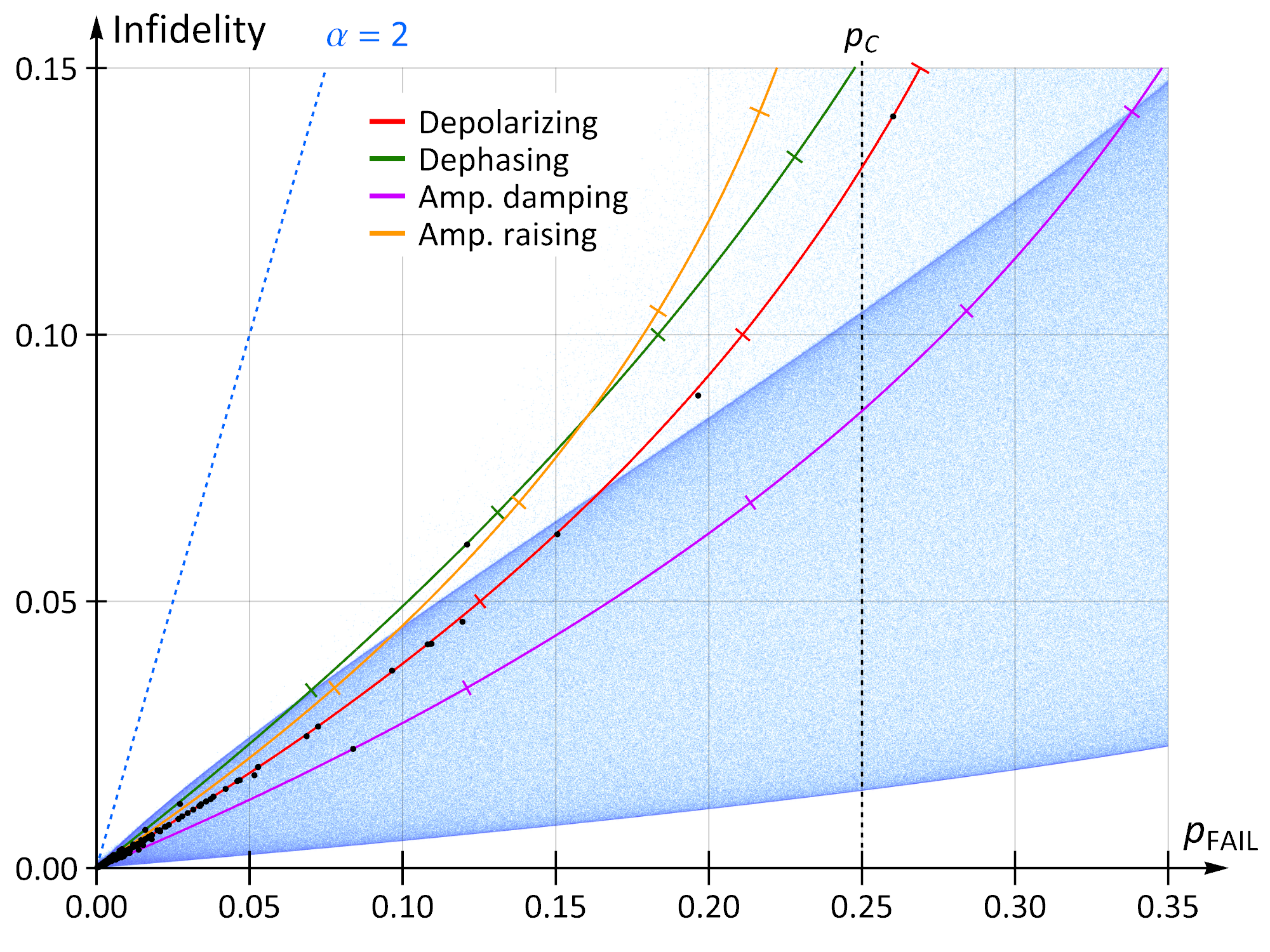}
\end{minipage}\hfill
\begin{minipage}[b]{0.49\textwidth}
\includegraphics[width=\textwidth]{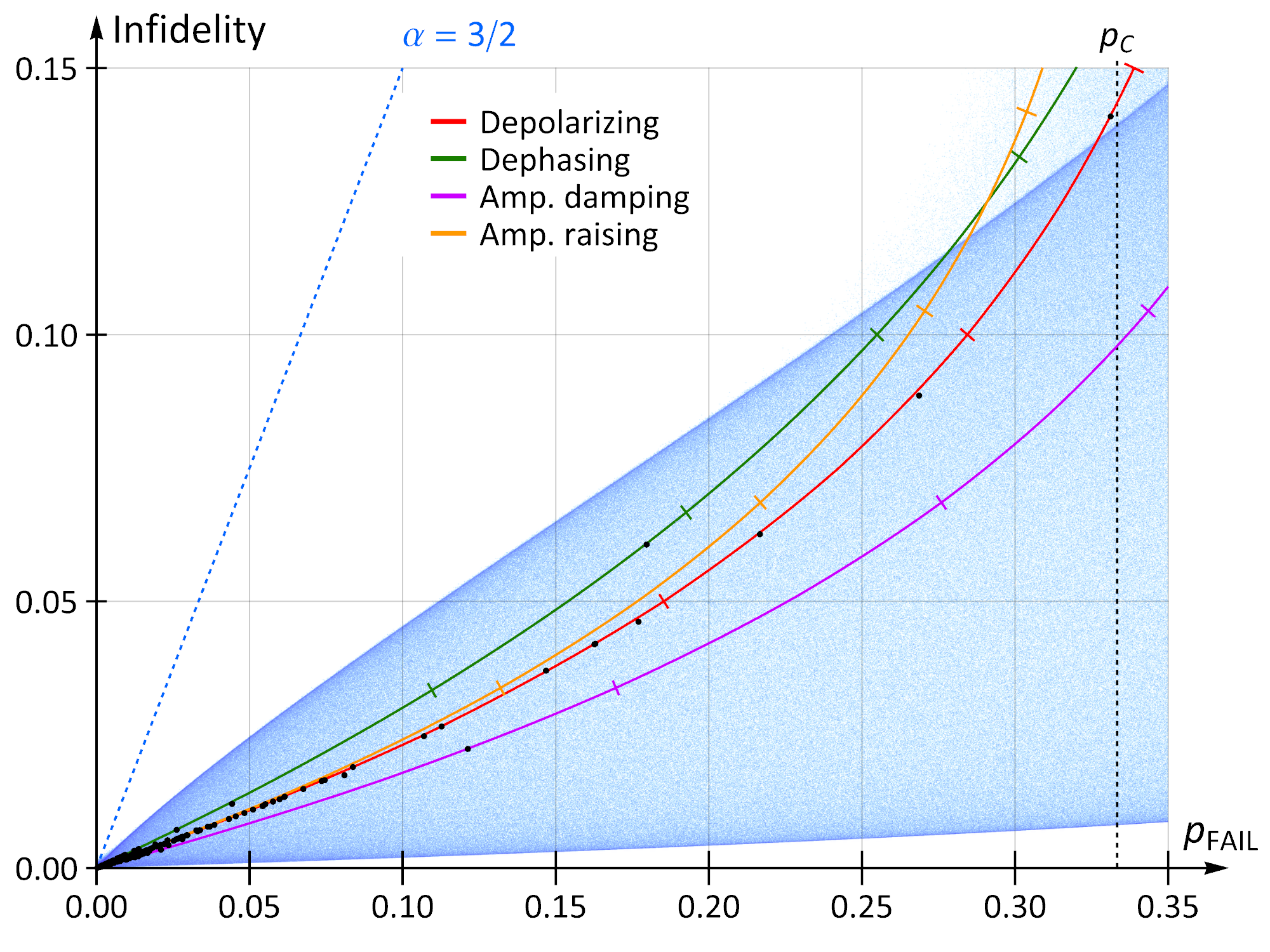}
\end{minipage}
\caption{Numerical survey of qubit channels for the promise problems $\mathtt{EO}^1_3$ (left) and $\mathtt{EO}^1_5$ (right). Blue regions each depict around $100$ million qubit channels, positioned according to their probability of failure and infidelity, with the noiseless $\sqrt{X}$ gate at the origin. The soundness guarantee of \ac{QSQ} can be seen as a sharp upper boundary near the origin. Blue dashed line shows a conservative worst-case infidelity upper bound, with slope $\alpha$ given by the optimal classical strategy with $\pfail = p_C$ and $\inFid = 1/2$. Curves show the behavior of the $\sqrt{X}$ gate under different noise models, with tick marks every 10\% increment for the noise parameter $t$. Black dots depict noisy $U_s$ channels for every available qubit of every simulated backend in IBM's quantum platform, placed according to their reported noise models.}\label{fig:cptp_survey}
\end{figure*}

\subsection{Numerical analysis for the S gate}

We investigate the \ac{QSQ} robustness for the restricted promise problems $\mathtt{EO}^1_\imax$ for $\imax=3$ and $\imax=5$,  which require a $4$-state \ac{pvDFA} to be solved, but which can be solved exactly by a qubit under the action of the $\sqrt{Z}$ gate. Here, we choose the target quantum model $(\ketbra{0}{0}, \{\sqrt{X}\}, \{\ketbra{a}{a}\}_{a \in \{0,1\}})$, with $\sqrt{X} = \kb{+}{+}+\ii\kb{-}{-}$ instead, which is equivalent to the \ac{QFA} solving the problem in \cref{eq:EO_s} under a unitary gauge. The minimum classical probability of failure $p_C$ is $1/4$ for $\imax=3$ and $2/3$ for $\imax=5$
 for to $\mu$ being uniform distribution (see \cref{app:proof_pfa} for details).

Since we want to investigate how the noise robustness of quantum gates changes with varying $\imax$, in what follows, we assume perfect state preparation and measurements in the implemented model, a quantum model of the form $(\ketbra{0}{0}, \{ \Lambda_s \}, \{\ketbra{a}{a}\}_{a \in \{0,1\}})$. 
In order to obtain tight upper bounds on the infidelity, the unitary gauge in \cref{eq:inFid} must be optimized globally, which generally poses a challenging nonconvex optimization problem.
Nevertheless, it is important to note that \cref{eq:inFid} gives a conservative worst-case upper bound for the infidelity, even if the (anti-) unitary found numerically is not optimal.
In particular, one finds that the quantum models realizing the optimal classical strategy have $\inFid = \tfrac{1}{2}$, which establishes a general worst-case infidelity upper bound $\inFid \le \alpha \, \pfail$ for a given $\pfail$ in this protocol, with $\alpha = \tfrac{1}{2 p_C}$.

The results of the numerical survey are shown in \cref{fig:cptp_survey}, together with the particular noise models described below. In the plots, the faint blue regions jointly depict around $200$ million randomly sampled quantum models, each a point $(\pfail,\inFid)$, covering every possible quantum behavior that can be certified in the protocol. The convergence of the models near the origin visually demonstrates the strong correlation between the infidelity of a model and its probability of failure in the \ac{QSQ} protocol. In particular, the soundness of \ac{QSQ} is evidenced by the sharp upper boundary of the blue regions near the origin. The dense triangular regions correspond predominantly to unitary channels, and as $\pfail$ increases, a faint distribution of non-unitary channels can also be seen above it, sloping upwards towards the classical strategy at $(p_C,\tfrac{1}{2})$.

The robustness of a \ac{QSQ} protocol to particular types of noise can be demonstrated by investigating their effects on $\pfail$ and $\inFid$ for a given promise problem. A noisy model can be obtained by applying a noise channel after the noiseless target gate, i.e., $\Lambda_\sigma = \mathcal{N}_t \circ U_\sigma$. Here, we study \ac{QSQ} robustness under the effect of common noise models, as well as the simulated noise models offered by IBM's quantum computing platform.

First, we consider the effects of depolarizing ($\mathcal{N}_t^\Delta$), dephasing ($\mathcal{N}_t^\Phi$), amplitude damping ($\mathcal{N}_t^0$), and amplitude-raising ($\mathcal{N}_t^1$) channels, with a noise parameter $t \in [0,1]$. They are defined as,
\begin{equation}
\begin{split}
    \mathcal{N}_t^\Delta(\rho) &= (1 - t) \rho + t \Tr[\rho] \frac{\one}{2}, \\
    \mathcal{N}_t^\Phi(\rho) &= (1 - \tfrac{t}{2}) \rho + \tfrac{t}{2} Z \rho Z, \\
    \mathcal{N}_t^{0}(\rho) &= K_0 \rho K_0^\dagger + K_1 \rho K_1^\dagger, \\
    \mathcal{N}_t^{1}(\rho) &= \tilde{K}_0 \rho \tilde{K}_0^\dagger + \tilde{K}_1 \rho \tilde{K}_1^\dagger,
\end{split}
\end{equation}
with $\tilde{K}_i = X K_i X$ and 
\begin{equation}
    K_0 = \begin{bmatrix}1 & 0 \\ 0 & \sqrt{1-t}\end{bmatrix}, \; K_1 = \begin{bmatrix}0 & \sqrt{t} \\ 0 & 0\end{bmatrix}. 
\end{equation}
The amplitude-raising channel realizes the optimal \ac{pvPFA} strategy for $t = 1$.
In addition, each noise curve could be obtained analytically through the computer algebra software Wolfram Mathematica~\cite{Mathematica}, allowing an asymptotic analysis for the low-noise regime $t \approx 0$ where the infidelity scales linearly as $\inFid = \alpha \, \pfail$, with $\alpha$ as given in \cref{tab:noise_slopes}. As observes, the scaling improves as the restricted promise problem length increases. 

\begin{table}[h!]
	\centering
	\begin{tabular}{l|ccc}
		\textbf{Noise channe}l & $\imax=3$ & $\imax=5$ \\
		\hline
		Depolarizing & 1/3 & 1/5 \\
		Dephasing & 4/9 & 4/15 \\
		Amplitude Damping & 8/33 & 8/51 \\
        Amplitude Raising & 8/21 & 8/39 \\
	\end{tabular}
	\caption{Asymptotic slopes $\alpha$ in the low-noise regime.}
	\label{tab:noise_slopes}
\end{table}

To investigate the robustness of \ac{QSQ} under more realistic noise, we use the simulated noise models offered by IBM's quantum computing platform under Qiskit~\cite{qiskit2024}. These models are generated using the calibration information of real IBM quantum processing units and are routinely used to mimic the behaviors of the actual physical hardware. For our analysis, we selected all IBM fake backends having $\sqrt{X}$ as the native gate, and extracted the corresponding noise channel for every available qubit. The $(\pfail,\inFid)$ values were computed for each of the resulting quantum models, and are depicted in \cref{fig:cptp_survey} as black dots.

Our numerical survey shows that the restricted promise problems of varying length have a rich structure for quantum gate certification. The longer the tested sequences get, the noise parameter resolution also increases. Finally, the overall slope for the worst-case noise also decreases from $\alpha=2$ to $\alpha=3/2$ as we change from $\imax=3$ to $\imax=5$.


\begin{table*}[t]
    \centering
    \begin{tabular}{l|ccccc}
        Promise Problem & $\abs{\Sigma}$ & \ac{QFA} operators & $ \abs{Q} $ & $\abs{S}$ & Ref.~within this work \\ 
        \hline
        $\mathtt{EO}^k$ & 1 & $\sqrt[2^k]{Z}$ & $2$ & $2^{k+1}$ &\cref{ex:EO}\\ 
        $\mathtt{DIOF}^k$ & $k$ &$\{Z,\sqrt{Z},\sqrt[4]{Z},\dots\}$& $2$ &$2^{k+1}$& \cref{ex:DIOF}\\
        $\mathtt{Cl}$ & $2$ & $\{\sqrt{Z},H\}$& $2$ & $6$ &\cref{ex:Cl}\\
        $\mathtt{N\text{-}EO}^k$ & $N$ & $\sqrt[2^k]{Z}^{(i)}$& $2^N$ & $2^{(k+1)N}$&\cref{ex:NEO}\\
        $\operatorname{\mathtt{GEO}^r_q}$ & 1 &$\sqrt[r]{Z_q}$ & $q$ & $q^{m_q+1} $&\cref{ex:GEO}\\
        $\mathtt{EO}^k_{\imax}$ & 1 &$\sqrt[2^k]{Z}$& $2$& $\min\{ 2^{k+1} , \imax+1\}$ & \cref{ex:EOimax}\\
    \end{tabular}
    \caption{Overview of the considered promise problems, their alphabet size, unitary operators of \ac{QFA}, minimal number of states of a \ac{QFA} and a \ac{pvDFA}, and the corresponding reference in the text.}
    \label{tab:PP-overview}
\end{table*}

\section{Summary and Outlook}~\label{sec:summary}

In this work, we argue that promise problems offer a natural framework for designing sound tests for sets of quantum gates.
We support this claim by showing that the tasks underlying the framework of quantum system quizzing~\cite{noller2025classical,noller2025sound} are instances of promise problems.
Interestingly, promise problems have previously been studied in the context of finite automata, where they serve to demonstrate separations between quantum and classical automata in terms of the required number of states, often exhibiting an exponential gap with respect to a parameter of a promise problem.
Here, we revisit several well-known examples of promise problems and discuss their relevance to gate-set certification.
We further introduce new promise problems, including ones naturally arising from gate sets used in quantum computing.
To strengthen the link between automata theory and gate certification, we discuss the structural uniqueness of automata that solve promise problems.
We also investigate how restricting the word length in promise problems affects the robustness of gate certification.

Finally, we argue that restricted promise problems can serve as novel but fruitful tasks in quantum information theory to demonstrate quantum advantage over classical counterparts. We give a brief example to spark future works: it can be shown that qubit entanglement-breaking channels can outperform classical \acp{pvPFA} in the same dimension for the simplest promise problem $\mathtt{EO}^1_2$. Here, the minimal probability of failure for a \ac{pvPFA} is $p_C=1/3$, while the entanglement-breaking channel, defined as follows
\begin{equation}
    \Lambda(\rho) = \frac{1}{2}\sum_{\tau\in\{\pm,\pm_y\}}(\bra{\tau}\rho\ket{\tau})S\ket{\tau}\bra{\tau}S^\dagger
\end{equation}
achieves the success probability of $\frac{1}{3}(1+5/8+17/32)\approx 0.712>2/3$. Note that the quantum entanglement breaking strategy can be further optimized. We note that similar quantum advantages with entanglement-breaking channels have also been reported in Ref.~\cite{vieira2024ebc} in a sequential scenario involving repeated measurements, although in that case, advantages were only observed beyond two-level systems and were harder to derive.

\begin{acknowledgements}
    We thank Oskari Kerppo, Martin Kliesch, Robert Koenig,  and Dan Brown for useful discussions.
    This research was funded by the Deutsche Forschungsgemeinschaft (DFG, German Research Foundation), project numbers 441423094, 236615297 - SFB 1119 and the Fujitsu Germany GmbH as part of the endowed professorship ``Quantum Inspired and Quantum Optimization''.
    This project was funded within the QuantERA II Programme that has received funding from the EU's H2020 research and innovation programme under the GA No 101017733. 
\end{acknowledgements}


\onecolumn

\section*{Appendix}
\appendix


\section{Minimal pvDFA solving $\mathtt{EO}^k$}\label{app:EO_proof}
\begin{proposition}
    A minimal \ac{pvDFA} that solves $\mathtt{EO}^k$ has $2^{k+1}$ states.
\end{proposition}
We translate the proof from  Ref.~\cite{bianchi2014complexity} to our notation.
\begin{proof}
    Since \acp{pvDFA} only allow deterministic transitions, i.e., a unique transition leaving each state, every unary \ac{pvDFA} consists of an initial ``tail'' with length $t$ before entering a ``loop'' with length $l$, such that $t + l = |S|$. $t$, as well as $l$, may be set to zero, although in the case of input strings of length larger than $|S|$, $l$ must be at least one to ensure well-defined outputs for such strings. Such a unary \ac{pvDFA} implements a modulo $l$ operation with offset $t$ on the length of the input string. Hence, $\mathtt{EO}^k$ is solved if a \ac{pvDFA} $\A$ can ensure 
    \begin{equation}\label{eq:modcond}
        \{ (i2^k - t)\bmod{l} \,|\, i \in \Znn, i\equiv 0\bmod 2 \} \cap
        \{ (i2^k - t)\bmod{l} \,|\, i \in \Znn, i\equiv 1\bmod 2 \} = \emptyset .
    \end{equation} 
First we show that a \ac{pvDFA} with $l = 2^{k+1}$ and $t=0$ indeed solves $\mathtt{EO}^k$. 

Since,
\begin{equation}\begin{split}
    \{ i\cdot2^k\bmod{2^{k+1}} \,|\, i \in \Znn, i\equiv 0\bmod 2 \} &= 
    \{j \cdot 2^{k+1} \bmod{2^{k+1}} \,|\, j \in \mathbb{Z}_{\geq 0} \} = \{0, 0, \dots\},  \\
    \{ i\cdot2^k\bmod{2^{k+1}} \,|\, i \in \Znn, i\equiv 1\bmod 2 \} &= 
    \{(j \cdot 2^{k+1} - 2^{k})\bmod{2^{k+1}} \,|\, j \in \mathbb{Z}_{\geq 0} \} = \{2^{k}, 2^{k}, \dots\},    
\end{split}\end{equation} 
\cref{eq:modcond} is clearly fulfilled and a \ac{pvDFA} $\mathcal{A}$ with $S_\a = \{s_0\}$ and  $S_\r = \{s_{2^k}\}$ categorizes the tested input strings correctly.

Next, let us study tail-free \acp{pvDFA}  with $|S| < 2^{k+1}$. First, consider the case of $l$ odd. The input sequence with length $l \cdot 2^{k}$ is contained in $\mathtt{EO}_\n^{k}$. Since $l \cdot 2^{k} \bmod{l} = 0$ but $0 \in \mathtt{EO}_\y^{k}$ always, this fails to solve the promise problem.
For $l$ even, we factorize $l$ into $l = 2^{n} \cdot p$, where $p$ consists of only odd factors. If $p = 1$, $l$ divides $2^k$ such that sequences from both sets will be sent to the accepting state. For $p \neq 1$, we observe that the sequence with length $p\cdot2^k = p \cdot 2^n \cdot 2^{k-n} \in \mathtt{EO}_\n^{k}$ which can be divided by $l$, causing a contradiction, too. Hence, no tail-free \acp{pvDFA}  with  $|S| < 2^{k+1}$ solve this promise problem.

In the following, we assume $t \neq 0$. The tail has two effects on the outputs of the \ac{pvDFA}. $t$ introduces a constant offset in the outputs of the \ac{pvDFA} $i \cdot 2^k - t \bmod l = i \cdot 2^k\bmod l - t \bmod l$ and secondly sends sequences with length $\leq t$ to distinct states which can be labeled freely. By applying similar argumentations as the tail-free case, we can show that there is no \ac{pvDFA} with $|S| < 2^{k+1}$ solving the promise problem.

First, consider $l$ odd. Then, $l \cdot 2^{k} \in \mathtt{EO}_\n^{k}$ and  $l \cdot 2^{k} - t \bmod l = - t \bmod l$. Likewise, we know that $2 \cdot l \cdot 2^{k} \in \mathtt{EO}_\y^{k}$ and $2\cdot l \cdot 2^{k} - t \bmod l = - t \bmod l$.
If $l$ even, we again consider $l = 2^{n} \cdot p$. Then, $p \cdot 2^n \cdot 2^{k-n} - t \bmod l = - t \bmod l$. However $2 \cdot p \cdot 2^n \cdot 2^{k-n} \in \mathtt{EO}_\y^{k}$ with same output $- t \bmod l$.
The smallest \ac{pvDFA} solving $\mathtt{EO}^k$ has indeed $|S| = 2^{k+1}$ with $l = 2^{k+1}$. A \ac{pvDFA} with tail $t \neq 0$ still requires a loop with $l = 2^{k+1}$.
\end{proof}

\section{Minimal pvDFA solving $\mathtt{Cl}$}
\label{app:DFA_sh}

\begin{proposition}
    A minimal \ac{pvDFA} that solves $\mathtt{Cl}$ has $6$ states.
\end{proposition}

\begin{lstlisting}[language=Python, caption={Restricted $\mathtt{Cl}$ to length up to 8.},label=lst:cl]
Cl_y = ['', 'hh', 'hsh', 'shs', 'hssh', 'hhhh', 'ssss', 'hhhsh', 'hhshs', 'hshhh', 'shhhs', 'shshh', 'hsssh', 'hshhsh', 'shhsss', 'hhhhhh', 'sshhss', 'sssshh', 'hhhssh', 'hshshs', 'hhssss', 'ssshhs', 'hssssh', 'hsshhh', 'shshsh', 'shsshs', 'shshhhh', 'hhhhhsh', 'hssshhh', 'sshshss', 'ssssshs', 'hhshshh', 'shsssss', 'hshssss', 'shshssh', 'hhhsssh', 'hsssssh', 'hhhshhh', 'hhhhshs', 'hsshhsh', 'hsshshs', 'ssshsss', 'hhshhhs', 'shhhshh', 'hshhhhh', 'hshhssh', 'sssshsh', 'shhhhhs', 'ssshhshh', 'shhhshsh', 'shshshhh', 'sshhshhs', 'sssshhhh', 'shhhsshs', 'hshhshhh', 'hshshshh', 'hssshhsh', 'hshhsssh', 'shhhhsss', 'shsshshh', 'shhssshh', 'hshhhshs', 'ssshhhhs', 'shsshhhs', 'hhhhhssh', 'hshshhhs', 'sssshssh', 'sshsshss', 'hhshshsh', 'hssssssh', 'hhshhsss', 'shhshhss', 'hhsssshh', 'hhhshshs', 'hssshshs', 'hsssshhh', 'hhssshhs', 'sshhhhss', 'ssssssss', 'shhsshhs', 'shshsssh', 'hsshhhhh', 'hhshsshs', 'hhhhhhhh', 'shshhshs', 'shshhhsh', 'hhhsshhh', 'hshhhhsh', 'hhhshhsh', 'hsshhssh', 'hsshssss', 'sshhsshh', 'hhhhssss', 'hhhssssh', 'hhsshhss']
Cl_n = ['ss', 'shhs', 'hhss', 'sshh', 'sshsh', 'ssshs', 'shsss', 'hshss', 'sshssh', 'hsshss', 'shhshh', 'sshhhh', 'hhhhss', 'hhsshh', 'ssssss', 'hhshhs', 'shhhhs', 'ssshhhs', 'hhhshss', 'shhhsss', 'hhshsss', 'hhssshs', 'shsshhs', 'hshhhss', 'shhshsh', 'sshhhsh', 'shshhss', 'sshhshs', 'hshshhs', 'hssshss', 'sshshhh', 'hshsshh', 'shssshh', 'hhsshsh', 'ssshshh', 'sshsssh', 'shhsshs', 'hhsshssh', 'sshhhhhh', 'ssshhsss', 'hhssssss', 'hshssshs', 'hsshsshh', 'shhshssh', 'hhsshhhh', 'sshshhsh', 'shhhhshh', 'hhhhhhss', 'ssshshsh', 'sshssssh', 'hsshhhss', 'hhhsshss', 'shhhhhhs', 'hhhhsshh', 'shsshsss', 'sshsshhh', 'hshsshsh', 'sshhssss', 'sshshshs', 'hsshshhs', 'hhshhshh', 'sshhhssh', 'hhhhshhs', 'shhsssss', 'ssshsshs', 'ssssshhs', 'sssshhss', 'shsssshs', 'shshshss', 'shssshsh', 'hhshhhhs', 'sssssshh', 'hsssshss', 'hshhshss', 'hshshsss', 'shhshhhh']
\end{lstlisting}

\begin{proof}
    The $6$-state \ac{pvDFA} depicted in \cref{fig:shautomaton} clearly solves the promise problem $\mathtt{Cl}$, since it is obtained from the mapping of a \ac{QFA} solving $\mathtt{Cl}$ (see~\cref{obs:map}).
    We demonstrate the minimality of this $\abs{S}=6$ \ac{pvDFA} solution by enumerating over all possible \acp{pvDFA} with $\abs{S}\leq 6$ and verifying whether each candidate \ac{pvDFA} solves a restricted promise problem, which includes all sequences up to length $8$ of $\mathtt{Cl}$ (see Listing~\ref{lst:cl}). This verification is conducted via a computer program due to the large search space.
    
    We observe that input sequences up to length $7$ exclude all \acp{pvDFA} with $\abs{S}\leq 5$ as valid solutions to $\mathtt{Cl}$, while including sequences of length $8$ excludes all with $\abs{S}=6$ \ac{pvDFA} except the one with the same structure as the \ac{pvDFA} in \cref{fig:shautomaton}. 
    By the latter, we mean that for each \ac{pvDFA} we can find a permutation of states which connects it to that \ac{pvDFA} solution.
\end{proof}

\section{Minimal pvDFA solving ${\mathtt{N\text{-}EO}^k}$}~\label{app:multiqubit}

\begin{proposition}
    A minimal \ac{pvDFA} that solves $\operatorname{\mathtt{N-EO}^k}$ has $2^{N(k+1)}$ states.
\end{proposition}

\begin{proof}
    First, we devise a \ac{pvDFA} with $2^{N(k+1)}$ that solves the promise problem. We notice that $\operatorname{\mathtt{N-EO_k}}$ can be solved by counting the number of each letter $\sigma_i$ in modulo $2^{k+1}$. The modulo counting is implemented by a cyclic transition of length $2^{k+1}$. Then, we can label each state with the current gate counts. To represent all possible gate counts of $N$ different gates in $\operatorname{\mathtt{N-EO_k}}$, we can use $N$ base $2^{k+1}$ numbers. If we include all possible $2^{N(k+1)}$ count instances, we can faithfully process all possible input sequences, even beyond the promised inputs.
    Next, we prove there is no smaller \ac{pvDFA} that solves the promise problem. If $|S|<2^{N(k+1)}$, there is at least one state $s$ that has at least two inequivalent count labels, $l_1,l_2$. Let's call the input sequence leading to $s$ with count $l_1$ and $l_2$, $p_1$ and $p_2$, respectively. We call the inverse of an input sequences $p \notin \operatorname{\mathtt{N-EO}^k_{00\dots0}}$ to be sequences such that  $p^{-1}\cdot p \in  \operatorname{\mathtt{N-EO}^k_{00\dots0}}$. Then, $p_1 \cdot p_2^{-1}$ is an input that reaches a state that is from the category $S_{00\dots0}$. However, generally $p_1 \cdot p_2^{-1} \notin \operatorname{\mathtt{N-EO}^k_{00\dots0}}$ causing a contradiction. The only case where this would be true is if $p_1 = p_2$, which we have excluded. The last step is to repeat the sequence $p_1 \cdot p_2^{-1}$ such that this generic sequence becomes compatible with the promised inputs. This can be achieved in the worst case with $2^k$ repetitions of $p_1 \cdot p_2^{-1}$. Note that it may happen that $2^k\cdot(p_1 \cdot p_2^{-1}) \in \operatorname{\mathtt{N-EO}^k_{00\dots0}}$. Then, repeat only $2^{k-1}$ times.
\end{proof}

\section{Minimal pvDFA solving $\mathtt{GEO}^r_{q}$ }~\label{app:GEORQ}
\begin{proposition}
    A \ac{pvDFA} that solves $\mathtt{GEO}^r_q$ requires at least $q^{m_q + 1}$ states,  where $m_q$ is the highest power of $q$, such that $q^{m_q}$ divides $r$.
\end{proposition}

\begin{proof}
    Analogous to the proof of \cref{res:EO}. If $r$ and $q$ are co-prime, the ${\bmod \,q}$ operation on the input sequence length is enough to differentiate $\mathtt{GEO}^r_{q,0}, \dots,\mathtt{GEO}^r_{q,q-1}$. This is no longer the case if $q$ occurs in the prime factors of $r$. Express $r = r'\cdot q^{m_q}$, we are looking for $l$ such that 
    \begin{equation}
        \bigcap_{j=0}^{q-1}\{ (n\cdot q + j)\cdot r'\cdot q^{m_q} \bmod{l}\} = \emptyset,
    \end{equation}
    pairwise disjoint for all $n\in \mathbb{N}$. Pairwise disjoint sets are given by $ l = q^{m_q + 1}$: 
    \begin{equation*}
      \bigcap_{i=0}^{q-1} \{ (n\cdot q + i)\cdot r'\cdot q^{m_q} \equiv i\cdot m\ (\bmod\; q^{m_q+1})\}  
    \end{equation*}
    Left to show is that $0 \not\equiv m \not\equiv \dots \not\equiv (q-1)\cdot m$. Note that $ a \equiv b \bmod c \iff k a \equiv k b \bmod k c$. Hence, we can reduce $m \equiv r'\cdot q^{m_q} (\bmod q^{m_q+1})$ to $m' \equiv r' (\bmod q)$. For multiples of $r'$, it holds
    \begin{align*}
        u \cdot r' \equiv v \cdot r'  \Rightarrow (u - v)\cdot r' \equiv 0 ,
    \end{align*}
    and $r'$ and $q$ co-prime by definition. For $u,v < q$, the only solution is $u = w$. Which shows that $\bmod{q^{m_q + 1}}$ solves the promise problem. Any smaller choice of the modulus $l$ causes a contradiction because $l\cdot r \equiv 0$ but $\sigma^{l\cdot r} \not\in \texttt{GEO}^r_0$ unless $q | l$, which would, however, map all sequences to $0$.
\end{proof}

\section{Minimal pvDFA solving $\mathtt{EO}^{k}_\imax$}~\label{app:KLEO}

\begin{proposition}
    A minimal \ac{pvDFA} that solves $\mathtt{EO}^k_\imax$ has $\min\{\imax+1,2^{k+1}\}$ states.
\end{proposition}

\begin{proof}
    Since $\mathtt{EO}^k_{\y,\imax} \subset \mathtt{EO}^k_\y$, and $\mathtt{EO}^k_{\n,\imax} \subset \mathtt{EO}^k_\n$, the \ac{pvDFA} with $|S| = 2^{k + 1}$ solves the promise problem as per~\cref{res:EO}. We consider the case $\imax < 2^{k+1}-2$. As before, examine the two input sets regarding their length per input sequence in ascending order
    \begin{align}
    \mathtt{EO}^{k}_{\imax,\y} &= \{0, 2\cdot2^k, \dots \} \\
        \mathtt{EO}^{k}_{\imax,\n} &= \{1\cdot2^k, 3\cdot2^k, \dots \}.
    \end{align}
    The promise problem is solved if the \ac{pvDFA} $\mathcal{A}$ maps the two distinct sets into distinct sets $\mathcal{A}(\mathtt{EO}^{k}_{\imax,\y}) \cap \mathcal{A}(\mathtt{EO}^{k}_{\imax,\n}) = \emptyset $.
    Let's assume first that $\mathcal{A}$ is a tail-free \ac{pvDFA} with cycle length $l=\imax+1$. Then,
    \begin{align}
        \mathcal{A}(\mathtt{EO}^{k}_{\imax,\y}) &= \{0, 2\cdot m, \dots \} \\
        \mathcal{A}(\mathtt{EO}^{k}_{\imax,\n}) &= \{1\cdot2^k \bmod l  =: m, 3\cdot m, \dots \}.
    \end{align} 
    These outputs induce another cycle on the \ac{pvDFA} states $0 \rightarrow 1\cdot m \rightarrow 2\cdot m \rightarrow 3 \cdot m \rightarrow \dots$, where $\mathtt{EO}^{k}_{\imax,\y}$-states and $\mathtt{EO}^{k}_{\imax,\n}$-states are visited alternately. If $l = 2\cdot 2^k$ and, for simplicity, $l \nmid 2^k$, we know that the first $l_o$ sequences can be correctly represented. Odd $l_o$ would cause a conflict at the $l_o+1$-th sequence.  Since this output cycle is defined on $l$ states with step size $m$, the cycle length is the smallest combination of $p,l_o$ such that $p\cdot l = l_o\cdot m$. This means for $\gcd(l,m)=g$ we have $p \cdot g \cdot l' = l_o \cdot g \cdot m'$, $l_o = l'$. Note that $l_o$ is always odd -- the only scenario of $l_o$ even would occur if $l$ is some even number $l = 2^i\cdot l'$ and $\gcd(l,m)$ odd or of lower multiplicity than $i$ of the prime factor $2$. However since $2^k = n \cdot l + m \Rightarrow 2^i(2^{k-i} + l') = m$ which means $2^i$ is a common divisor of both $m$ and $l$ and thus $l_o$ odd. Hence, such a \ac{pvDFA} can only accommodate a promise problem with up to $l_o$ sequences.
    
    Next, we will see that in case $l_o \neq l$, we can increase the maximum sequence number to $|S|$ by adding a tail of length $t = 1$, such as $l = |S| - 1 $. From $2^k = n \cdot g \cdot l' + g \cdot m' = n\cdot g (l' + m')$ we see that $\gcd(l,m)$ can only be power of $2$. Hence, if  $\gcd(l,m) \neq 1 \Rightarrow \gcd(l-1, \tilde{m}) = 1$. Since the single tail-state can correctly recognize the empty string, such a \ac{pvDFA} can account for $|S|$ sequences in total. Thus, a promise problem with $\imax < 2^{k+1} - 1$ can be solved with a \ac{pvDFA} with $|S| = \imax$ states.
    \end{proof}

\begin{corollary}
If $\imax\leq 2^{k+1}$ and $\imax$ is odd, it is enough to consider
\begin{equation}\begin{split}
          \mathtt{EO}^{k}_{\imax,\y} &= \{ \sigma^{0}, \sigma^{2 \cdot 2^k}\} \\
         \mathtt{EO}^{k}_{\imax,\n} &= \{ \sigma^{i2^k} | i \in \Znn, i\equiv 1\bmod 2, i \leq \imax\}
\end{split}\end{equation}
and if $i_{max}$ even, then it is enough to consider
\begin{equation}\begin{split}
          \mathtt{EO}^k_{\imax,\y} &= \{ \sigma^{i2^k} | i \in \Znn, i\equiv 0\bmod 2, i \leq \imax \}   \\
          \mathtt{EO}^{k}_{\imax,\n} &= \{ \sigma^{1 \cdot 2^k} \}.
\end{split}\end{equation}
\end{corollary}

\begin{proof}
    In the proof of \cref{theorem:finitepp}, we established that $\mathtt{EO}^{k}_{\imax}$ for appropriately small $\imax$ can be solved with a \ac{pvDFA} with $|S| = \imax+ 1$ states because only the $\imax+ 2$-th sequence would cause a contradiction in the \ac{pvDFA} outputs since it will map onto the same state as the first sequence from the opposite input set. This also means that omitting sequences from the opposite input set larger than the shortest non-empty sequence would not change the required \ac{pvDFA} dimension. Note that for the odd $i_{max}$, only keeping the non-empty sequence would allow for a two-state \ac{pvDFA} to solve the promise problem, since the problem reduces to differentiating between an empty sequence and a non-empty sequence.
\end{proof}

\section{Optimal success probability of pvPFAs for $\mathtt{EO}^1_\imax$}\label{app:proof_pfa}

In order to certify the presence of quantum advantages, \ac{QFA} failing probability must be lower than $p_C$ for two-dimensional \acp{pvPFA}.

Using \cref{eq:pfa_accprob,eq:pfail} and the $\mathtt{EO}^1_\imax = (\mathtt{EO}^{1}_{\imax,\y},\mathtt{EO}^{1}_{\imax,\n})$ promise problem, the probability of failure for its two-state \ac{pvPFA} with a single input symbol can be written as:

\begin{equation}\label{eq:pfa_pfail}
\pfail = 1 - \sum_{a\in\LL}\sum_{w \in A_a} \mu(w) \Pr[a|w] = 1 - \frac{1}{\imax + 1} \sum_{k=0}^{\imax} m_{a_k} T^{2k} \pi, \qquad T = \begin{bmatrix}x & y \\ 1 - x & 1 - y\end{bmatrix}
\end{equation}
where the distribution $\mu$ is taken to be uniform, and $a_k = ``\y"$ for $k$ even, and $a_k = ``\n"$ for odd, and $x, y \in [0,1]$ are two transition probabilities to be optimized over. Thus, we wish to find $p_C = \inf_{x,y} \pfail$. By convexity, and without loss of generality, we may choose the initial distribution $\pi = [1, 0]^\T$. For $m_a$, we have four possibilities: either $m_\y = [1,0]$,   $m_\y = [0,1]$, $m_\y = [1,1]$, or $m_\y = [0,0]$,  with $m_\n = 1 - m_\y$, corresponding to the choice of whether to start on the accept state. It is easy to see that $m_\y = [1,0]$ always provides an advantage due to the empty word term achieving $\Pr[a = \y|\eps] = 1$. Thus, it suffices to optimize over the two parameters $x,y \in [0,1]$, where the global optimum can be found to be at $x = y = 0$, corresponding to a \ac{pvDFA} which always transitions to the non-starting state for $\imax=3$ and $\imax=5$. The optimal classical probability of failure is then simply
\begin{equation}
    p_C = 1 - \frac{\Pr[a=\y|\epsilon] + |\mathtt{EO}^{1}_{\imax,\n}|}{|\mathtt{EO}^{1}_{\imax,\y}| + |\mathtt{EO}^{1}_{\imax,\n}|} = 1 - \frac{1 + |\mathtt{EO}^{1}_{\imax,\n}|}{\imax + 1} = \frac{\imax - \ceil{\imax/2}}{\imax + 1}.
\end{equation}

\section{Details on numerical optimization}\label{app:numerics}

We wish to obtain the infidelity between the target quantum model $\QM = (\ketbra{0}{0}, \{\sqrt{X}\}, \{\ketbra{a}{a}\}_{a \in \{0,1\}})$ and a model $\tilde{\QM} = (\ketbra{0}{0}, \{ \Lambda_\iS \}, \{\ketbra{a}{a}\}_{a \in \{0,1\}})$, minimized over a unitary gauge. As global phases are irrelevant, we use the $3$-parameter unitary
\begin{equation}U = \begin{bmatrix}
 \e^{\ii (\alpha + \beta )} \cos \theta  & \e^{\ii (\alpha -\beta )} \sin \theta  \\
 -\e^{-\ii (\alpha - \beta )} \sin \theta  & \e^{-\ii (\alpha +\beta )} \cos \theta  \\
\end{bmatrix},
\end{equation}
such that the state infidelity reduces to $\inFid^U_\rho = 1 - \Tr[U \kb{0}{0} U^\dagger \kb{0}{0}] = 1 - \cos^2 \theta = \sin^2 \theta$. On the other hand, from the spectrum of $U M_a U^\dagger - M_a$ it can be shown that $\mathrm{dist}^U_M = |\sin \theta|$, such that the contribution to the infidelity from the measurement dominates that of the state for any $\theta$. The overall infidelity for a given $U$ thus reduces to $\max\;\{\; \inFid^U_\text{gate},\; \abs{\sin \theta}\}$.

We sample random channels $\Lambda_\iS$ through a variation of the Kraus operator method as described in Ref.~\cite{kukulski2021}, where instead of generating random matrices from the complex Ginibre ensemble, we generate matrices where each entry is generated according to $r^2 \e^{\ii \theta}$ for uniformly sampled $r \in [0,1)$ and $\theta \in [0,2\pi)$. This was done to improve the sampling at the boundaries of the region populated by channels in the $\inFid$ vs.~$\pfail$ plots, as otherwise the majority of sampled channels would concentrate near the classical $\pfail \ge p_C$ region. For each channel, the unitary gauge was optimized via the Luus–Jaakola method~\cite{luusjaakola1973} using unitaries sampled over the Haar measure. The optimization code was written in \texttt{C++} using the Armadillo library~\cite{sanderson2025armadillo}, with the source code being available upon request.
\newline
\twocolumn
\bibliographystyle{unsrtnat}
\bibliography{citations.bib}

\end{document}